\documentclass[copyright]{eptcs}
\providecommand{\event}{PDMC 2009} 
\providecommand{\volume}{??}
\providecommand{\anno}{2009}
\providecommand{\firstpage}{1}
\providecommand{\eid}{??}
\usepackage{breakurl} 

\title{Distributed Branching Bisimulation Minimization\\by Inductive Signatures}
\author{Stefan Blom \qquad \qquad Jaco van de Pol
\institute{University of Twente, Formal Methods and Tools
\thanks{This work has been partially funded by the EU under grant number FP6-NEST STREP 043235 (EC-MOAN).}
\\P.O.-box 217, 7500 AE, Enschede, The Netherlands
           \email{\{sccblom,vdpol\}@cs.utwente.nl}}
}
\def\titlerunning{Inductive Branching Signatures}
\def\authorrunning{S.C.C. Blom, J.C. van de Pol}

\usepackage{version}
\excludeversion{report}
\includeversion{pdmc}

\usepackage{ifthen}
\usepackage{amssymb}
\usepackage{arrows}
\usepackage{tikz}                       
\usetikzlibrary{arrows,automata,positioning}
\tikzset{every picture/.style={->,>=stealth',initial text={}}}
\tikzset{every state/.style={minimum size=3mm,inner sep=1mm}}

\newcommand{\nat}{\mathbb{N}}
\newtheorem{definition}{Definition}
\newtheorem{example}[definition]{Example}
\newtheorem{lemma}[definition]{Lemma}

\newtheorem{proposition}[definition]{Proposition}
\newenvironment{proof}[1][Proof.~]{\noindent {\bf #1}}{\mbox{}\hfill$\Box$\medskip}

\newcommand{\bisim}{\mathrel{\underline{\leftrightarrow}}}
\newcommand{\canonical}[1]{#1\!\!\downarrow}

\usepackage{listings}
\lstdefinelanguage{pseudo}{
keywords={proc,while,do,end,for,return,if,then,to,true,false,and,hashtable,or,until,with,in,send,recv,set,int,take,from,add,and,or},
mathescape,
escapeinside={(@}{@)}
}

\begin{document}
\maketitle

\begin{abstract}
We present a new distributed algorithm for state space minimization modulo branching bisimulation. Like its predecessor it uses signatures for refinement, but the refinement process and the signatures have been optimized to exploit the fact that the input graph contains no $\tau$-loops.

The optimization in the refinement process is meant to reduce both the number of iterations needed
and the memory requirements. In the former case we cannot prove that there is an improvement, but
our experiments show that in many cases the number of iterations is smaller.
In the latter case, we can prove that the worst case memory use of the new algorithm is linear in the size of the state space, whereas the old algorithm has a quadratic upper bound.

The paper includes a proof of correctness of the new algorithm and the results of a number of experiments that compare the performance of the old and the new algorithms.

\begin{report}
This report is an extension of \cite{sigmin-inductive-pdmc} with full proofs.
\end{report}
\end{abstract}



\section{Introduction}

The idea of distributed model checking of very large systems, is to
store the state space in the collective memory of a cluster of
workstations, and employ parallel algorithms to analyze the graph. One
approach is to generate the graph in a distributed way, and on-the-fly
(i.e. during generation) run a distributed model checking algorithm.
This is what is done in the DiVinE toolset \cite{DBLP:conf/cav/BarnatBCMRS06}.
This is useful if the system is expected to contain bugs, because
the generation can stop after finding the first bug.

Another approach is to generate the full state space in a distributed
way, and subsequently run a distributed bisimulation reduction
algorithm. The result is usually much smaller, and satisfies the same
temporal logic properties. The minimized graph could be small enough
to analyse with sequential model checkers. This approach is useful
for certification, because many properties can be checked on the
minimized graph. This paper contributes to the second approach.

\medskip

The process-algebraic way of abstracting from actions is to hide them
by renaming them to the invisible action $\tau$. To reason about equivalence of
these abstracted models, branching bisimulation \cite{GlWe96,DBLP:journals/ipl/Basten96}~
can be used. Because branching bisimulation is coarser than strong bisimulation,
this leads to smaller state spaces modulo reduction.

Distributed minimization algorithms have been proposed in
\cite{DBLP:journals/sttt/BlomO05,DBLP:journals/sttt/BlomO05a} for
strong bisimulation, and in \cite{DBLP:journals/entcs/BlomO03} for
branching bisimulation. These are signature-based algorithms, which
work by successively refining the trivial partition, according to the
(local) signature of states with respect to the previous partition.

The best-known sequential algorithm~\cite{DBLP:conf/icalp/GrooteV90}
for branching bisimulation reduction assumes that the state space
has no $\tau$-cycles. The idea is that any $\tau$-cycles can be
removed in linear time, by Tarjan's algorithm to detect (and eliminate)
strongly connected components (SCC)~\cite{DBLP:journals/siamcomp/Tarjan72}.
Eliminating SCCs preserves branching bisimulation.

Because eliminating $\tau$-cycles in distributed graphs seemed
complicated, the algorithm in \cite{DBLP:journals/entcs/BlomO03} works
on {\em any} LTS, i.e. it doesn't assume the absence of
$\tau$-cycles. This generality came with a certain cost: signatures
have to be transported over the transitive closure of silent
$\tau$-steps.
\footnote{A $\tau$-step is silent if the source and destination are
equivalent (with respect to the previous partition).}
For some cases this leads to increased time and memory usage.

Later, several distributed SCC detection (and elimination) algorithms
have been developed~\cite{OP0501,Simona,McLendonIII2005901,BCP09}. It has
already been reported in~\cite{Simona} that running SCC elimination as
a preprocessing step to the branching minimization algorithm
of~\cite{DBLP:journals/entcs/BlomO03}, reduces the overall time.
Note that this gain was achieved even though the minimization algorithm
doesn't assume that the input graph is $\tau$-acyclic.

In this paper, we further improve this method, by exploiting the fact
that the input graph of the minimization algorithm has no
$\tau$-cycles. Using this extra knowledge, we are able to develop a
distributed minimization algorithm that runs in less time and memory.

At the heart of our improved method is a notion of {\em inductive
  signature}.  Normally, during a round of signature computations,
only the signatures of the previous round may be used. The basic idea
of inductive signatures is that the new signature of a state may
depend on the {\em current} signature of its
$\arrowsuper{a}$-successors, provided $a$ is guaranteed to
terminate. We will first illustrate this notion for strong
bisimulation, and then apply it to branching bisimulation, where
$\tau$ is cycle-free, i.e.~$\arrowsuper{\tau}$ is a terminating transition.
Note that if all action labels are terminating, the graph is actually a directed
acyclic graph, for which it is known that there is a linear algorithm for bisimulation
reduction.

{\bf Overview.} In the next section, we will explain the theory
and prove the correctness of the improved signature bisimulation.
In section \ref{algorithm section}, we explain how we turned the
definition of inductive signature bisimulation onto a distributed algorithm
and how we implemented it on top of the LTSmin
toolset\footnote{\url{http://fmt.cs.utwente.nl/tools/ltsmin/}}.
We show the results of running the tool on several problems in Section~\ref{experiments}.


\section{Theory}

In this section, we start by recalling the basic definitions of LTS and bisimulation.
Followed by the definitions of signature refinement from previous papers.
Then we present inductive signatures for strong bisimulation followed by
inductive signatures for branching bisimulation. We end this section with the
correctness proof for branching bisimulation.

\subsection{Preliminaries}
\label{preliminaries subsection}

First, we fix a notation for
labeled transition systems and recall the definitions
of strong bisimulation and  branching bisimulation \cite{GlWe96,DBLP:journals/ipl/Basten96}.
Our transition systems are labeled with actions from a given
set {\sf Act}. The invisible action $\tau$ is a member of ${\sf Act}$.

\begin{definition}($LTS$)
A labeled transition system ($LTS$)
is a triple $(S,\to,s^0)$, consisting of a set of states
$S$, transitions $\to \subseteq S \times {\sf Act} \times S$
and an initial state $s^0\in S$.
\end{definition}

We write $s \arrowsuper{a} t$ for $(s,a,t)\in \to$, and use 
$\arrowsuper{a}^*$ to denote the transitive reflexive closure of $\arrowsuper{a}$.

Both strong and branching bisimulation can be defined in two ways. As a relation between
two LTSs or as a relation on one LTS. We choose the latter.

\begin{definition}[strong bisimulation]
Given an LTS $(S,\to,s^0)$. A symmetric relation $R \subseteq S \times S$ is a strong bisimulation
if \(
\forall s,t,s' \in S : ~\forall a \in {\sf Act} : ~ s\mathrel{R}t \wedge s \arrowsuper{a} s'
\Rightarrow
\exists t' \in S : ~ t \arrowsuper{a} t' \wedge s' \mathrel{R} t'
\).
\end{definition}

\begin{definition}[branching bisimulation]
Given an LTS $(S,\to,s^0)$. A symmetric relation $R \subseteq S \times S$ is a branching bisimulation
if
\begin{pdmc}
\[
\forall s,t,s' \in S : ~ \forall a \in {\sf Act} : ~ s\mathrel{R}t \wedge s \arrowsuper{a} s'
\Rightarrow
( a\equiv\tau \wedge s' \mathrel{R} t )
\vee
(\exists t',t''\in S : ~ t \arrowsuper{\tau}^* t' \wedge s \mathrel{R} t' \wedge t' 
\arrowsuper{a} t'' \wedge s'  \mathrel{R} t'' )
\]
\end{pdmc}
\begin{report}
\[
\forall s,t,s' \in S : ~ \forall a \in {\sf Act} : ~ s\mathrel{R}t \wedge s \arrowsuper{a} s'
\Rightarrow
\left\{\begin{array}{c}
( a\equiv\tau \wedge s' \mathrel{R} t )
\\
\vee
\\
(\exists t',t''\in S : ~ t \arrowsuper{\tau}^* t' \wedge s \mathrel{R} t' \wedge t' 
\arrowsuper{a} t'' \wedge s'  \mathrel{R} t'' )
\end{array}\right.
\]
\end{report}
Two states $s,t\in S$ are branching bisimilar (denoted $s \bisim t$) if there exists
a branching bisimulation $R$ such that $s \mathrel{R} t$.
\end{definition}

For proving correctness, we will use a few properties:

\begin{proposition}
Given an LTS:
\\[0.2em]
\phantom{X}~~~$\bullet$ the relation $\bisim$ is a branching bisimulation;
\\[0.2em]
\phantom{X}~~~$\bullet$ if $R$ is a branching bisimulation then $R \subseteq \bisim$.
\end{proposition}

For a proof see \cite{GlWe96}.

\smallskip

To talk about bisimulation reduction algorithms, we need the terminology of partition refinement.
Given a set $S$.
\\[0.2em]
\phantom{X}~~~$\bullet$ A set of sets $\{S_1 , \cdots , S_N\}$ is a {\em partition} of $S$ if
$S = S_1 \cup \cdots \cup S_N$ and $\forall i \neq j: S_i \cap S_j = \emptyset$.
Each set $S_i$ is referred to as a {\em block} and must be non-empty.
\\[0.1em]
\phantom{X}~~~$\bullet$  A partition $\{S_1 , \cdots , S_N\}$ is a {\em refinement} of a  partition
$\{S_1' , \cdots , S_M'\}$ if $\forall i \exists j: S_i \subseteq S_j'$.
\\[0.1em]
\phantom{X}~~~$\bullet$  Any partition $\{S_1 , \cdots , S_N\}$ can be represented with an
identity function $ID : S \to \nat$, defined as $ID(s)=i$, if $s\in S_i$.

\subsection{Signature Refinement}

We continue with the previously published variant of signature refinement.
Because many results are correct for finite LTSs only, we assume that both ${\sf Act}$
and all LTSs are finite for the remainder of the paper.

The signature of a state is computed with respect to a partition.
Intuitively, the signature of a state is the set of possible moves (actions) that are possible in
a state with respect to the partition (represented by a number). Formally:

\begin{definition}\label{defsig}\mbox{ }
\begin{itemize}
\item The set of signatures ${\sf Sig}$ is the set of {\em finite subsets} of ${\sf Act} \times \nat$.
\item A partition $\pi$ of an LTS $(S,\to,s^0)$ is a function $\pi:S \to \nat$.
\item A signature function is a function $sig : (S \to \nat) \times S \to {\sf Sig}$, such that
for all isomorphisms $\phi :\nat \to \nat$ and all partitions $\pi$:
\[
\forall s \in S : sig(\phi \circ \pi,s)=\{ (a, \phi(n)) \mid (a,n) \in sig(\pi,s)\}
\]
\end{itemize}
\end{definition}

The last clause is to ensure that the equality on signatures is independent of how numbers are
chosen to represent partitions. This is important because we want to do a refinement process,
where based on a partition, we compute signatures, which we turn into a partition, for which we compute signatures, etc. until the partition is stable. This requires translating signatures (or better
pairs of previous partition numbers and signatures) to integers, which we do by means of given
isomorphisms:
\[
h_1,h_2,\cdots: \nat \times {\sf Sig} \to \nat \enspace .
\]
These isomorphisms exist due to the fact that signatures are finite, which implies that
the set of signatures is countable. The actual refinement process works as follows:

\begin{itemize}
\item Given an initial partition $\pi_0$ of $S$.
\item Given a signature function $sig$.
\item Define
\(
\pi_{i+1} (s) = h_{i+1} ( \pi_i(s) , sig(\pi_i,s) )
\)
\item Define the relation $\pi_i \subseteq S \times S$ as
\(
s \mathrel{\pi_i} t \mbox{, if } \pi_i(s) = \pi_i(t) \enspace.
\)
\item There exists $N\in\nat$ such that the relation
$\pi_N = \pi_{N+1}$. Define $\pi_0^{sig}=\pi_N$.
\end{itemize}

Note that although the definitions of the functions $\pi_{i+1}$ depend on the choice of the isomorphisms $h_{i+1}$, the relations $\pi_i$ will be the same regardless of the choice of $h_{i+1}$, due to the third clause of Definition~\ref{defsig}.
This definition is turned into an algorithm by starting with $\pi_i$ for $i=0$, and computing $\pi_{i+1}$ from $\pi_i$ until
the partition is stable ($\pi_{i+1}\equiv\pi_i$).

\smallskip

For the computed refinement to make sense, we need notions of signatures that correspond to meaningful equivalences.
For example, the signatures of a state according to strong bisimulation and branching bisimulation are

\begin{definition}[classic signatures]
\[
\begin{array}{lll}
{\rm sig_s}(\pi,s) & = & \{ ( a , \pi(t) ) \mid s \arrowsuper{a} t \}
\\
{\rm sig_b}(\pi,s) & = & \{ ( a , \pi(t) ) \mid s \arrowsuper{\tau} s_1 \cdots \arrowsuper{\tau} s_n \arrowsuper{a} t ,
\pi(s)=\pi(s_i) \wedge (a \neq \tau \vee \pi(s)\neq\pi(t)) \}
\end{array}
\]
\end{definition}

The signature of a state says which equivalence classes are reachable from the state by performing an action.
For example in strong bisimulation, if there is an $a$ step from a state $s$ to a state $t$ then
the equivalence class of $t$ is reachable by means of an $a$ step form $s$ which is expressed by putting the pair $(a,\pi(t))$ in the signature of $s$.

The case for branching bisimulation is more complicated. The set of actions includes the {\em invisible action} $\tau$.
The intent of this label is that whatever happens is unimportant. Thus $\tau$ steps are ignored, except if
they change the {\em branching behaviour}. An ignored $\tau$ step is called {\em silent}.
More formally a $\tau$ step is silent with respect to a partition if it is between states in the same equivalence class.

See \cite{DBLP:journals/sttt/BlomO05} and \cite{DBLP:journals/entcs/BlomO03} for more explanation.

\subsection{Inductive signatures for strong bisimulation}

In the classical definition of the strong bisimulation signature,
the signatures depend on the previous partition only.
One may wonder if in some cases the current partition can be used. The answer is yes.
If for each label you consistently use the old partition or consistently use the new partition
then it still works. Of course if we use the current partition then we must ensure that all
signatures are well defined. This is ensured if the subgraph of edges for which we use the
current partition is acyclic. This is guaranteed if we have a {\em well-founded partition}
of the set of actions.
A well-founded partition is a partition $A_?,A_>$ of the set of actions, such that the relation
$\{ (s,t) ~|~ s \arrowsuper{a} t \wedge a\in A_{>}\}$ is well-founded:

\begin{definition}
A pair $\langle A_? , A_> \rangle$ is a well founded partition of ${\sf Act}$ for an LTS $(S,\to,s^0)$
if $A_? \cap A_> = \emptyset$, $A_? \cup A_> = {\sf Act}$ and the LTS is $A_>$ cycle free.
The order $> \subseteq S\times S$ is defined by $> \equiv \big(\cup_{a \in A_>} \arrowsuper{a}\big)^+$.
\end{definition}

Based on the well-founded order $>$ we can give inductive definitions and proofs. 
For example, we can define inductive strong bisimulation signatures:

\begin{definition}[inductive strong bisimulation]
Given an LTS $(S,\to,s^0)$, a well founded partition $\langle A_? , A_> \rangle$ for it,
an initial partition function $\pi_0:S \to \nat$ and isomorphisms $h_1,h_2,\cdots : \nat \times {\sf Sig \to \nat}$.
Define
\[
\begin{array}{lll}
sig_{i+1}(s) & = &
\{ (a,\pi_i(t)) \mid s \arrowsuper{a} t \wedge a \in A_? \} \cup
\{ ( a , \pi_{i+1}( t ) ) \mid s \arrowsuper{a} t \wedge a \in A_> \}
\\
\pi_{i+1} ( s ) & = & h_{i+1} ( \pi_i(s) , sig_{i+1} ( s ) )
\end{array}
\]
\end{definition}

Note that $sig_{i+1}(s)$ is defined inductively in terms of any
$\pi_i$-values, and only $\pi_{i+1}$ values of states that are smaller in $>$.
To show how the definition works and how the choice of the partition influences performance,
we continue with an example.

\begin{example}
Consider the following LTS:
\begin{center}
\begin{tikzpicture}[baseline=(0)]
\node[state,initial] (0) {$0$};
\node[state] (1) [right=of 0] {$1$};
\node[state] (2) [right=of 1] {$2$};
\node[state] (3) [right=of 2] {$3$};
\node[state] (4) [right=of 3] {$4$};
\node[state] (5) [right=of 4] {$5$};
\path
  (0) edge [bend left=30] node[above] {$a$} (1)
  (1) edge [bend left=30] node[above] {$a$} (2)
  (2) edge [bend left=30] node[above] {$a$} (3)
  (3) edge [bend left=30] node[above] {$a$} (4)
  (4) edge [bend left=30] node[above] {$a$} (5)
  (5) edge [bend left=30] node[below] {$b$} (4)
  (4) edge [bend left=30] node[below] {$b$} (3)
  (2) edge [bend left=30] node[below] {$b$} (1)
  (1) edge [bend left=30] node[below] {$b$} (0);
\end{tikzpicture}
\end{center}
If we take $A_> := \{a\}$, and set $\pi_0(s):=0$ for all states,
we get the following run:
\[\begin{array}{rcll}
sig_1(5) & := & \{(b,0)\}       & \pi_1(5) = 1 \\
sig_1(4) & := & \{(b,0),(a,1)\} & \pi_1(4) = 2 \\
sig_1(3) & := & \{(a,2)\} & \pi_1(3) = 3 \\
sig_1(2) & := & \{(b,0),(a,3)\} & \pi_1(2) = 4 \\
sig_1(1) & := & \{(b,0),(a,4)\} & \pi_1(1) = 5 \\
sig_1(0) & := & \{(a,5)\} & \pi_1(0) = 6 \\
\end{array}\]
Note that every state got a different signature, so in this case
we reach the final partition in one round. Also note that the
order of computation was completely fixed, because the label $a$ imposes
a total order on the states.

Next, consider the same example, but let $A_> = \{b\}$. Note that
this is also terminating. Again,
we take $\pi_0(s)=0$ for any state $s$.
\[\begin{array}{rcll@{~~~~~,~~~~~~}rcll}
sig_1(0) & := & \{(a,0)\}       & \pi_1(0) = 1 &
sig_1(3) & := & \{(a,0)\}       & \pi_1(3) = 1 \\
sig_1(1) & := & \{(a,0),(b,1)\} & \pi_1(1) = 2 &
sig_1(4) & := & \{(a,0),(b,1)\} & \pi_1(4) = 2 \\
sig_1(2) & := & \{(a,0),(b,2)\} & \pi_1(2) = 3 & 
sig_1(5) & := & \{(b,2)\}       & \pi_1(5) = 4 \\
\hline
sig_2(0) & := & \{(a,2)\}       & \pi_2(0) = 5 &
sig_2(3) & := & \{(a,2)\}       & \pi_2(3) = 5 \\
sig_2(1) & := & \{(a,3),(b,5)\} & \pi_2(1) = 6 &
sig_2(4) & := & \{(a,4),(b,5)\} & \pi_2(4) = 7 \\
sig_2(2) & := & \{(a,1),(b,6)\} & \pi_2(2) = 8 &
sig_2(5) & := & \{(b,7)\}       & \pi_2(5) = 9 \\
\hline
sig_3(0) & := & \{(a,6)\}       & \pi_3(0) = 10 &
sig_3(3) & := & \{(a,7)\}       & \pi_3(3) = 11 \\
sig_3(1) & := & \{(a,8),(b,10)\} & \pi_3(1) = 12 &
sig_3(4) & := & \{(a,9),(b,11)\} & \pi_3(4) = 13 \\
sig_3(2) & := & \{(a,5),(b,12)\} & \pi_3(2) = 14 &
sig_3(5) & := & \{(b,13)\}       & \pi_3(5) = 15 \\
\end{array}\]
Note that this time we need three iterations, but there is some
room for parallel computation, because the signature of $0$
and $3$ can be computed independently, because they
have no $b$ successors.
\end{example}

\subsection{Inductive signatures for branching bisimulation}

In the splitting procedure of the Groote-Vaandrager algorithm, whenever a state
has one or more $\tau$ successors inside the block that is being split, the algorithm
tests if the behavior of one of those $\tau$ successors includes all of the behavior of the state.
If such a successor exists, then the state is put in the same block as that successor.
Because of this splitting procedure the graph has to be $\tau$-cycle free.
A similar effect can be achieved by exploiting $\tau$ cycle freeness
when we define the branching signature. Thus, we
assume that $\tau\in A_>$ for all partitions $\langle A_?,A_> \rangle$.

The inductive branching signature is computed in two steps.
First, the $pre$-signature is computed, which consists of all transitions to all successors, including $\tau$-steps
to possibly equivalent states. Second, we look for a $\tau$-successor in the same block of the previous partition which contains all $pre$ behavior except the $\tau$ step to that successor.
If such a successor is found then the signature is the signature of that successor,
otherwise the signature is the $pre$-signature:

\begin{definition}[inductive branching bisimulation]
Given an LTS $(S,\to,s^0)$, a well founded partition $\langle A_? , A_> \rangle$ for it with $\tau\in A_>$
and an initial partition function $\pi_0:S \to \nat$.
Define
\[
\begin{array}{lll}
pre_{i+1}(s) & = & \{ (a,\pi_i(t)) \mid s \arrowsuper{a} t \wedge a \in A_? \} \cup
\{ ( a , \pi_{i+1}( t ) ) \mid s \arrowsuper{a} t \wedge a \in A_> \}
\\
sig_{i+1}(s) &=& \mbox{if there exists a $t$ with $s \arrowsuper{\tau} t$,
 $\pi_i(s)=\pi_i(t)$ and $pre_{i+1}(s) \subseteq sig_{i+1}(t)\cup \{(\tau,\pi_{i+1}(t))\}$}
 \\
 && \mbox{~~~then } sig_{i+1}(t)
 \\
 && \mbox{~~~else } pre_{i+1}(s)
\\
\pi_{i+1} ( s ) & = & h_{i+1} ( \pi_i(s) , sig_{i+1} ( s ) )
\end{array}
\]
\end{definition}

It is not immediately obvious that this is well-defined: what if there exists more
than one $\tau$-successor that passes the test? The answer is: then they have the same signature.
We prove this in lock step with the observation that if a signature $\sigma$ contains a pair $(a,n)$,
then any state with signature $\sigma$ has a path of silent $\tau$ steps to a state where an $a$ step is possible to a final state in partition $n$.

To avoid unnecessary case distinctions between $a \in A_?$ and $a\in A_>$, we introduce the notation
\[
\hat{a}\stackrel{\rm def}{=}\left\{\begin{array}{ll}0&\mbox{, if } a \in A_?
\\
1&\mbox{, if } a \in A_>\end{array}\right.
\]
This allows us to abbreviate ``$\pi_i(s)$ if $a\in A_?$ and
$\pi_{i+1} (s)$ if $a \in A_>$'' by $\pi_{i+\hat{a}}(s)$.
\begin{pdmc}
Due to space restrictions, we only sketch the essentials of the proofs.
Full proofs can be found in \cite{sigmin-inductive-report}.
\end{pdmc}

\begin{proposition}\label{well-defined prop}
For all states $s$:
\begin{enumerate}
\item  If there exist $t_1,t_2$ with $s \arrowsuper{\tau} t_1$, 
$s \arrowsuper{\tau} t_2$,
 $\pi_i(s)=\pi_i(t_1)=\pi_i(t_2)$, $pre_{i+1}(s) \subseteq sig_{i+1}(t_1)\cup \{(\tau,\pi_{i+1}(t_1))\}$
 and $pre_{i+1}(s) \subseteq sig_{i+1}(t_2)\cup \{(\tau,\pi_{i+1}(t_2))\}$ then $sig_{i+1}(t_1)=sig_{i+1}(t_2)$.
\item If $(a,n)\in sig_{i+1} ( s )$ then
$\exists s_1,\cdots, s_m, t : s \arrowsuper{\tau} s_1 \cdots \arrowsuper{\tau} s_m \arrowsuper{a} t
\wedge \pi_i(s)=\pi_i(s_j) \wedge n = \pi_{i+\hat{a}}(t)$.
\end{enumerate}
\end{proposition}
\begin{proof}
\begin{pdmc}
We prove both parts at once by induction on $\arrowsuper{\tau}^*$.
\\
Given a state $s$, we prove part 1 by contradiction.
Suppose that $sig_{i+1}(t_1)\neq sig_{i+1}(t_2)$.
\\
By definition $\{ (\tau,\pi_{i+1}(t_1)) ,  (\tau,\pi_{i+1}(t_2))\}\subseteq pre_{i+1}(s)$.
This implies that $(\tau,\pi_{i+1}(t_1)) \in  sig_{i+1}(t_2)$ and
$(\tau,\pi_{i+1}(t_2)) \in  sig_{i+1}(t_1)$. By using part 2,
we construct an infinite path
$s \arrowsuper{\tau} t_1 \equiv s_1 \arrowsuper{\tau}^+ s_1' \arrowsuper{\tau}^+ s_2 \arrowsuper{\tau}^+ \cdots $
with $\pi_{i+1}(s_i)=\pi_{i+1}(t_1)$ and $\pi_{i+1}(s_i')=\pi_{i+1}(t_2)$. This infinite path contradicts the cycle freeness.
\\
The proof of part 2 is elementary.
\end{pdmc}
\begin{report}
We prove both parts at once by induction on $\arrowsuper{\tau}^*$.
\\
Given a state $s$, we prove part 1 by contradiction.
Suppose that $sig_{i+1}(t_1)\neq sig_{i+1}(t_2)$. Then
\[
\{ (\tau,\pi_{i+1}(t_1)) ,  (\tau,\pi_{i+1}(t_2))\}\subseteq pre_{i+1}(s)
\]
and therefore:
\[
(\tau,\pi_{i+1}(t_1)) \in  sig_{i+1}(t_2) \mbox{ and } (\tau,\pi_{i+1}(t_2)) \in  sig_{i+1}(t_1)
\]
Let $s_1 = t_1$. Because $s \arrowsuper{\tau} t_1$, the induction hypothesis applies to
$t_1$. Thus by applying part 2, there exists a state $s_1'$, such that $s_1 \arrowsuper{\tau}^+
s_1'$ and $\pi_{i+1}(s_1')=pi_{i+1}(t_2)$. This implies that $sig_{i+1}(s_1')=sig_{i+1}(t_2)$.
So we can find $s_2$, such that $s_1' \arrowsuper{\tau}^+ s_2$ and $\pi_{i+1}(s_2)=\pi_{i+1}(t_1)$.
In other words we get an infinite sequence
\[
s_1 \arrowsuper{\tau}^+ s_1 \arrowsuper{\tau}^+ s_2 \arrowsuper{\tau}^+ \cdots
\]
In a finite state space this implies the existence of a $\tau$ cycle. Contradiction.
\\
Part 2  is proven by case distinction. We have two cases:
\begin{description}
\item[$sig_{i+1}(s)=pre_{i+1}(s)$] If $(a,n)\in pre_{i+1}(s)$ then for some $t$: $s \arrowsuper{a} t$
and $n = \pi_{i+\hat{a}}(t)$.
\item[$s \arrowsuper{\tau} t \wedge \pi_i(s)=\pi_i(t) \wedge sig_{i+1}(s)=sig_{i+1}(t)$]
By induction hypothesis, we have a sequence $t \arrowsuper{\tau} t_1 \arrowsuper{\tau} \cdots t_m \arrowsuper{a} t'$ satisfying the requirement for $t$. Which means that the requirement for $s$ is
satisfied by
\[
s \arrowsuper{\tau} t \arrowsuper{\tau} t_1 \arrowsuper{\tau} \cdots t_m \arrowsuper{a} t'
\]
\end{description}
\end{report}
\end{proof}

We will show how the new definition works and is different from the approach of \cite{DBLP:journals/entcs/BlomO03}, by means of an example:

\begin{example}
Consider the following three examples.  We have only drawn the nodes of the graphs which are relevant.
Let $\pi_0(s)=0$ for all $s$ and $\pi_i(s)=0$ for all nodes $s$ which have been omitted.
\begin{center}
\begin{tabular}{ccc}
\begin{tikzpicture}[node distance=1.5em,baseline=(1)]
\node[state,initial] (0) {$0$};
\node[state] (1) [above right=of 0] {$1$};
\node[state] (2) [below right=of 0] {$2$};
\node (3) [right=of 1] {};
\node (4) [right=of 2] {};
\node (5) [below=of 0] {};
\path
  (0) edge node[left, near end] {$\tau$} (1)
  (0) edge node[left, near end] {$\tau$} (2)
  (1) edge node[right] {$\tau$} (2)
  (0) edge node[left] {$a$} (5)
  (1) edge node[above] {$a,b$} (3)
  (2) edge node[above] {$a,b,c$} (4);
\end{tikzpicture}
&
\begin{tikzpicture}[node distance=1.5em,baseline=(1)]
\node[state,initial] (0) {$0$};
\node[state] (1) [above right=of 0] {$1$};
\node[state] (2) [below right=of 0] {$2$};
\node (3) [right=of 1] {};
\node (4) [right=of 2] {};
\node (5) [below=of 0] {};
\path
  (0) edge node[left, near end] {$\tau$} (1)
  (0) edge node[left, near end] {$\tau$} (2)
  (0) edge node[left] {$a$} (5)
  (1) edge node[above] {$a,b$} (3)
  (2) edge node[above] {$a,b,c$} (4);
\end{tikzpicture}
&
\begin{tikzpicture}[node distance=1.5em,baseline=(2)]
\node[state,initial] (0) {$0$};
\node[state] (1) [right=of 0] {$1$};
\node[state] (2) [above right=of 1] {$2$};
\node[state] (3) [below right=of 1] {$3$};

\node (4) [above right=of 0] {};
\node (5) [below right=of 0] {};
\node (6) [right=of 2] {};
\node (7) [right=of 3] {};

\path
  (0) edge node[left, near end] {$a$} (4)
  (0) edge node[left, near end] {$b$} (5)
  (0) edge node[above] {$\tau$} (1)
  (1) edge node[left, near end] {$\tau$} (2)
  (1) edge node[left, near end] {$\tau$} (3)
  (2) edge node[above] {$a$} (6)
  (3) edge node[below] {$b$} (7);
\end{tikzpicture}
\end{tabular}
\end{center}

\noindent Let $A_>=\{\tau\}$. Then for the left-most LTS on the left, we get:
\[
\begin{array}{rcll}
pre_1(2) & := & \{(a,0),(b,0),(c,0)\} &  \\
sig_1(2) & := & \{(a,0),(b,0),(c,0)\} & \pi_1(2) = 1
\\
pre_1(1) & := & \{(a,0),(b,0),(\tau,1)\} & Note: \{(a,0),(b,0),(\tau,1)\}\subseteq \{(a,0),(b,0),(c,0),(\tau,1)\}\\
sig_1(1) & := & \{(a,0),(b,0),(c,0)\} & \pi_1(1) = 1
\\
pre_1(0) & := & \{(a,0),(\tau,1)\} & Note: \{(a,0),(\tau,1)\}\subseteq \{(a,0),(b,0),(c,0),(\tau,1)\}\\
sig_1(0) & := & \{(a,0),(b,0),(c,0)\} & \pi_1(0) = 1
\end{array}\]
Note that $|dom(sig_1)|=|dom(sig_0)| = 1$, so $sig_1$ is stable, and all $\tau$-steps are silent.
\\[0.5ex]
For the middle LTS, we obtain:
\[\begin{array}{rcll}
pre_1(2) & := & \{(a,0),(b,0),(c,0)\} &  \\
sig_1(2) & := & \{(a,0),(b,0),(c,0)\} & \pi_1(2) = 1 \\

pre_1(1) & := & \{(a,0),(b,0)\} & \\
sig_1(1) & := & \{(a,0),(b,0)\} & \pi_1(1) = 2 \\

pre_1(0) & := & \{(a,0),(\tau,1),(\tau,2)\} & Note: 
\{(a,0),(\tau,1),(\tau,2)\}\not\subseteq \{(a,0),(b,0),(c,0),(\tau,1)\},\\
&&& \phantom{Note: }
\{(a,0),(\tau,1),(\tau,2)\}\not\subseteq \{(a,0),(b,0),(\tau,2)\}\\
sig_1(0) & := & \{(a,0),(\tau,1),(\tau,2)\} &  \pi_1(0) =  3
\end{array}\]
Note that $|dom(sig_1)|=3$, which cannot increase, so again $sig_1$ is stable.
In this case, none of the $\tau$-steps is silent.
\\[0.5ex]
For the LTS on the right, we get
\[\begin{array}{rcll@{~~~~~,~~~~~~}rcll}
sig_1(2) & := & \{(a,0)\} & \pi_1(2) = 1 &
sig_1(3) & := & \{(b,0)\} & \pi_1(3) = 2 \\
sig_1(1) & := & \{(\tau,1),(\tau,2)\} & \pi_1(1) = 3 &
sig_1(0) & := & \{(a,0),(b,0),(\tau,3)\} & \pi_1(0) = 4 \\
\end{array}\]
Already after one iteration it is detected that none of the $\tau$-steps is silent.
In the original definition in \cite{DBLP:journals/entcs/BlomO03}, this would be detected later, as the following example shows.
\[\begin{array}{r@{\,}c@{\,}ll@{~~~~~,~~~~~~}r@{\,}c@{\,}ll}
sigb_1(2) & := & \{(a,0)\} & \pi_1(2)=1 &
sigb_1(3) & := & \{(b,0)\} & \pi_1(3)=2 \\
sigb_1(1) & := & \{(a,0),(b,0)\} & \pi_1(1)=3 &
sigb_1(0) & := & \{(a,0),(b,0)\} & \pi_1(0)=3\\
\hline
sigb_2(2) & := & \{(a,0)\} & \pi_2(2)=1 &
sigb_2(3) & := & \{(b,0)\} & \pi_2(3)=2 \\
sigb_2(1) & := & \{(\tau,1),(\tau,2)\} & \pi_2(1)=4 &
sigb_2(0) & := & \{(a,0),(b,0),(\tau,1),(\tau,2)\} & \pi_2(0)=5\\
\end{array}\]
Note the two differences between inductive and classic signatures.
First, the fact that $0\arrowsuper{\tau}1$ is not silent is detected in
the first iteration by inductive and the second by classic signatures.
Second, in the inductive case the size of the signature is limited by
the number of outgoing transitions in the classic case it is not.
\end{example}

\subsection{Correctness}

We use the same proof technique as in previous work. That is,
we prove that bisimilar states are always in the same block
and that if a $\pi_i$ partition is stable ($\pi_i$ and $\pi_{i+1}$
denote the same relation) then $\pi_i$ is a bisimulation.
Thus because $\bisim$ is the coarsest bisimulation, we must have
that $\pi_i$ coincides with $\bisim$.
\begin{pdmc}
Again, we include proof sketches only. Full proofs are available in
\cite{sigmin-inductive-report}.
\end{pdmc}

In this section we work on a given LTS $(S,\to,s^0)$ and
well-founded partition $(A_?,A_>)$, with $\tau\in A_>$. We consider inductive branching bisimulation
and we let $s \bisim_i t$ denote $\pi_i(s) = \pi_i(t)$.

One of the properties of a $\tau$-cycle free LTS is that given a state
one can always follow $\tau$ steps to bisimilar states, until a state
is found that has no such step. These states are called canonical:

\begin{definition}
A state $s$ is \emph{canonical} (denoted $\canonical{s}$) if $\neg\exists s' :~ s \arrowsuper{\tau} s' \wedge s \bisim s'$.
\end{definition}

Canonical states have the important property that all visible behavior is present
as an immediate step rather than as a sequence of one or more invisible steps followed by a
visible step.

\begin{lemma}\label{preservation lemma}
If ${\bisim}\subseteq{\bisim_i}$ then
for all states $s,t$ we have $(s\bisim t \wedge \canonical{t}) \Rightarrow s \bisim_{i+1} t$
\end{lemma}

To prove this, we need two properties.
\begin{proposition}\label{two observations prop}
For all states $s,t$, we have
\\[0.em]
\phantom{2}~~~~1. $pre_{i+1}(s) \subseteq sig_{i+1}(s)\cup \{ ( \tau , \pi_{i+1}(s) ) \}$.
\\
\phantom{1}~~~~2. $pre_{i+1}(s) \subseteq sig_{i+1}(s)\cup \{ ( \tau , \pi_{i+1}(s) ) \}$.
\begin{pdmc}
\end{pdmc}
\begin{report}
\begin{enumerate}
\item $pre_{i+1}(s) \subseteq sig_{i+1}(s)\cup \{ ( \tau , \pi_{i+1}(s) ) \}$.
\item if $pre_{i+1}(s) = pre_{i+1}(t)$ then $sig_{i+1}(s) = sig_{i+1}(t)$.
\end{enumerate}
\end{report}
\end{proposition}
\begin{proof}
\begin{pdmc}
By distinguishing cases depending on which branch was taken in the if-then-else of the definition of inductive signature.
\end{pdmc}
\begin{report}
Both parts are proven by case distinction.
\begin{description}
\item[$\exists s' : s \arrowsuper{\tau} s' \wedge pre_{i+1}(s) \subseteq sig_{i+1}(s')\cup \{ ( \tau , \pi_{i+1}(s') )$]
By definition we have $sig_{i+1}(s)=sig_{i+1}(s')$. Thus, part 1 is trivial.
\\
It also means that $(\tau,\pi_{i+1}(s')) \in pre_{i+1}(s)$.
So $(\tau,\pi_{i+1}(s') \in pre_{i+1}(t)$. So for some $t'$, we have
$t \arrowsuper{\tau} t'$ and $\pi_{i+1}(t')=pi_{i+1}(s')$.
This implies that $sig_{i+1}(t')=sig_{i+1}(s')$. Finally, it follows that
$sig_{i+1}(t)=sig_{i+1}(s)$.
\item[otherwise]
By definition we have $sig_{i+1}(s)=pre_{i+1}(s)$. Thus, part 1 is trivial.
\\
Due to symmetry we have have that $sig_{i+1}(t)=pre_{i+1}(t)$ as well.
\end{description}
\end{report}
\end{proof}

\noindent\begin{proof}[Proof of Lemma \ref{preservation lemma}.]
\begin{pdmc}
By induction on the order $(s,t)\geq(s',t')$ iff $s\geq s'\wedge t\geq t'$.

Because any transition in $s$ is either matched by a transition
of $t$, or it is a silent $\tau$ step, we have
\\[0.2em]
\centerline{\(
pre_{i+1}(s) \subseteq pre_{i+1}(t) \cup \{ ( \tau , \pi_{i+1}(t) ) \}
\)}
\\[0.2em]
Now, we distinguish on whether $s$ is canonical or not.
\begin{itemize}
\item $\canonical{s}$: In this case $pre_{i+1}(s)=pre_{i+1}(t)$,
due to the fact that bisimilar canonical states have the same transitions.
This implies $sig_{i+1}(s)=sig_{i+1}(t)$ and thus $s \bisim_{i+1} t$.
\item $s\arrowsuper{\tau}s' \wedge s \bisim s'$:
By induction hypothesis $sig_{i+1}(s')=sig_{i+1}(t)$. Thus
\\[0.2em]
\centerline{\(
pre_{i+1}(s) \subseteq pre_{i+1}(t) \cup \{ ( \tau , \pi_{i+1}(t) ) \}
\subseteq sig_{i+1}(t) \cup \{ ( \tau , \pi_{i+1}(t) ) \}
= sig_{i+1}(s') \cup \{ ( \tau , \pi_{i+1}(s') ) \}
\)}
\\[0.2em]
and therefore $sig_{i+1}(s)=sig_{i+1}(s')$.
\end{itemize}
\end{pdmc}
\begin{report}
By induction on the order $\geq$ on pairs of states, defined as 
$(s,t)\geq(s',t')$ iff $s\geq s'\wedge t\geq t'$.

First, we prove that
\begin{equation}\label{subset claim}
pre_{i+1}(s) \subseteq pre_{i+1}(t) \cup \{ ( \tau , \pi_{i+1}(t) ) \}
\end{equation}
The elements of $pre_{i+1}(s)$ fit one of two cases:
\begin{itemize}
\item $(a,\pi_i(s'))$, for $s \arrowsuper{a} s' \wedge a \in A_?$:
Because $s \bisim t$ and $a\neq \tau$, we have $t \arrowsuper{\tau}^* t'' \arrowsuper{a} t'$ with $
s\bisim t'' \wedge s' \bisim t'$. Because $\canonical{t}$, we have $t''=t$.
By assumption $s' \bisim_i t'$ and thus $(a,\pi_i(s'))\in pre_{i+1}(t)$.
\item $(a,\pi_{i+1}(s'))$, for $s \arrowsuper{a} s' \wedge a \in A_{>}$:
We have three sub-cases:
\begin{itemize}
\item $t \arrowsuper{a} t'$ with $s' \bisim t'$:
By induction hypothesis $s' \bisim_{i+1} t'$ and thus $(a,h(sig_{i+1}(s'))\in pre_{i+1}(t)$.
\item $t \arrowsuper{\tau}^+ t'' \arrowsuper{a} t'$ with $s \bisim t'' \wedge s' \bisim t'$:
Impossible due to $\canonical{t}$.
\item $a=\tau$ and $s' \bisim t$:
By induction hypothesis $s' \bisim_{i+1} t$, so $(a,\pi_{i+1}(s'))=( \tau , \pi_{i+1}(t) )$.
\end{itemize}
\end{itemize}
This completes the proof of (\ref{subset claim}). Now, we distinguish on
whether $s$ is canonical or not.
\begin{itemize}
\item $\canonical{s}$: In this case we claim $pre_{i+1}(s)=pre_{i+1}(t)$.
Each of the inclusion is proven similar to the proof of (\ref{subset claim}) above.
This implies $sig_{i+1}(s)=sig_{i+1}(t)$ and thus $s \bisim_{i+1} t$.
\item $s\arrowsuper{\tau}s' \wedge s \bisim s'$:
By induction hypothesis $sig_{i+1}(s')=sig_{i+1}(t)$. Thus
\[
pre_{i+1}(s) \subseteq pre_{i+1}(t) \cup \{ ( \tau , \pi_{i+1}(t) ) \}
\subseteq sig_{i+1}(t) \cup \{ ( \tau , \pi_{i+1}(t) ) \}
= sig_{i+1}(s') \cup \{ ( \tau , \pi_{i+1}(s') ) \}
\]
Thus $sig_{i+1}(s)=sig_{i+1}(s')$.
\end{itemize}
\end{report}
\end{proof}

\begin{lemma}\label{bisimulation lemma}
If for all $s,t$: $s \bisim_i t \Leftrightarrow s \bisim_{i+1} t$ then
$\bisim_i$ is a branching bisimulation.
\end{lemma}
\begin{proof}
\begin{pdmc}
Corollary of Prop.\ref{well-defined prop}, part 2.
\end{pdmc}
\begin{report}
Suppose that $s \bisim_i t$ and $s \arrowsuper{a} s'$.
\\
We distinguish two cases:
\begin{itemize}
\item $a=\tau \wedge s \bisim_i s'$: This implies $s' \bisim_i t$.
\item $a\neq\tau \vee s \not{\bisim}_i s'$: This implies
$(a,\pi_{i+\hat{a}}(s'))\in sig_{i+1}(s)$. So $(a,\pi_{i+\hat{a}}(s'))\in sig_{i+1}(t)$.
If $(a,\pi_{i+\hat{a}}(s'))\in pre_{i+1}(t)$ then $t \arrowsuper{a} t'$ with $s' \bisim_i t'$.
Otherwise, there must be a $t^\circ$, such that $t \arrowsuper{a} t^\circ$
and $sig_{i+1}(t)=sig_{i+1}(t^\circ)$. By repeating
the case distinction we can construct $t \arrowsuper{\tau}^* t'' \arrowsuper{a} t'$ with
$s\bisim_i t'' \wedge s' \bisim_i t'$.
\end{itemize}
\end{report}
\end{proof}

\section{Distributed Algorithm}
\label{algorithm section}

In this section, we present a distributed algorithm for computing the branching bisimulation
equivalence relation.

\medskip

The input to the algorithm is an LTS $(S,\to,s^0)$, a well
founded partition $\langle A_? , A_> \rangle$,
and a function $owner : S \to \{ 1 , \cdots , W \}$ where $W$ is the number of workers.
The owner function is a given distribution of states among the workers.

The given isomorphisms of the theory are replaced by global hash tables in the implementation.
Each worker stores an equal part of this global hash table.The worker where the (new) ID of the pair
(oldID,signature) is stored is given by the second owner function
$owner : ID\times Sig \to \{1,\cdots,W\}$.

In the actual implementation states and edges are numbered entities. Since the theory
assumes that edges are triples, we need to introduce some new notation. Moreover, we have to
distinguish which worker owns which state and which edge, so we need some notation for
that as well.

\smallskip
The functions $src,dst$ and $lbl$  provide access to the source state, destination state and label
of an edge, respectively:
\\[0.2em]
\centerline{\(
\forall e\equiv(s,a,t) \in {\to} :
src(e)=s
\mbox{, }
lbl(e)=a
\mbox{ and }
dst(e)=t
\enspace .
\)}
\\[0.2em]
Each worker owns a set of states and needs to know the outgoing $\tau$ edges, $A_?$ edges and
$A_>$ edges:
\\[0.2em]
\centerline{\(
\begin{array}{lll@{~~~~~~~~~~~~~~}lll}
S_w &=& \{ s \in S \mid owner(s) = w \}
&
E^\tau_w &=& \{ e \in {\to} \mid src(e) \in S_w \wedge lbl(e)=\tau \}
\\
E_w^? &=& \{ e \in {\to} \mid src(e) \in S_w \wedge lbl(e)\in A_?\}
&
E_w^> &=& \{ e \in {\to} \mid src(e) \in S_w \wedge lbl(e)\in A_> \}
\end{array}
\)}
\\[0.2em]
Finally, we need the definitions of successor
and predecessor edges of a state:
\\[0.2em]
\centerline{\(
succ(s)=\{ e \mid src(e)=s \}
~~~~~~~~~~~~~~~~
pred(s)=\{ e \mid dst(e)=s \}
\)}
\\[0.2em]
Each worker stores both ingoing and outgoing edges of the states it owns in a way that allows it
to quickly enumerate the successors and predecessors of every state. 

\begin{table}[tp]\caption{Pseudo code for worker $w$ (inductive branching bisimulation reduction)}
\label{refinement listing}
\hspace*{1.5em}
\begin{minipage}{\textwidth}
\begin{lstlisting}[numbers=left,language=pseudo]
set sig[$S_w$], dest_sig[$E_w^\tau$], old_queue, sig_queue, new_queue
int old_id[$S_w$], current_id[$S_w$], dst_old[$E_w^?\cup E_w^\tau$], dst_new[$E_w^>$]
proc reduce()
   int old_count:=0 , new_count:=1 (@ \label{main init begin} @)
   for (@$t \in S_w $@) do current_id[t]:=0 end (@ \label{main init end} @)
   while old_count$\neq$new_count do
      old_count:=new_count;  indexed_set_clear() (@ \label{loop init begin} @)
      for (@$t \in S_w $@) do old_id[$t$]:=current_id[$t$]; current_id[$t$]:=$\perp$ end
      for $e$ in $E_w^? $ do dst_old[$e$]:=$\perp$ end ; for $e$ in $E_w^>$ do dst_new[$e$]:=$\perp$ end
      for $e$ in $E^\tau_w$ do dst_sig[$e$]:=$\perp$ ; dst_old[$e$]:=$\perp$ end
      old_queue := $S_w$; sig_queue:= $\{ s \in S_w \mid \neg\exists a,t:~s \arrowsuper{a} t\}$ ; new_queue := $\emptyset$ (@ \label{loop init end} @)
      do
      :: take $s$ from old_queue => (@ \label{old send begin} @)
            for $e$ in $pred(s)$ with $lbl(e)\in Act_? \cup \{ \tau \}$ do
               send set_old(e,old_id[$s$]) to owner($src(e)$) end (@ \label{old send end} @)
      :: recv set_old(e,id) => dst_old[e]:=id; check_ready($src(e)$) (@ \label{old recv begin} @) (@ \label{old recv end} @)
      :: take $s$ from sig_queue => (@ \label{sig begin} @)
            sig := compute_sig(s);
            for $e$ in $pred(s)$ with $lbl(e)=\tau$  do
               send set_sig(e,sig) to owner(src(e)) end
            send get_global($s$,old_id[s],sig) to owner(old_id[s],sig) (@ \label{sig end} @)
      :: recv set_sig(e,e_sig) => dest_sig[e] := e_sig; check_ready($src(e)$) (@ \label{sig recv begin} \label{sig recv end} @)
      :: recv get_global($s$,id_old,sig) => (@ \label{hash serve begin} @)
            send set_global(s,indexed_set_put(id_old,sig)) to owner(s) (@ \label{hash serve end} @)
      :: recv set_global(s,id) => current_id[s]:=id; add $s$ to new_queue (@ \label{hash recv begin}\label{hash recv end}  @)
      :: take $s$ from new_queue => (@ \label{new send begin}@)
            for $e$ in $pred(s)$ with $lbl(e)\in Act_>$  do
               send set_new($e$,current_id[$s$]) to $owner(src(e))$ end (@ \label{new send end}@)
      :: recv(set_new($e$,id)) => dst_new[$e$]:=id; check_ready($src(e)$) (@ \label{new recv begin}\label{new recv end} @)
      until $\forall s \in S:$ current_id[$s$] $\neq$ $\perp$
      new_count:=distributed_sum(index_count)
   end
end
\end{lstlisting}
\end{minipage}
\end{table}

\newcommand{\seelisting}[2]{\ifthenelse{\equal{\ref{#2 begin}}{\ref{#2 end}}}
{(See table \ref{#1 listing}, line \ref{#2 begin}.)}
{(See table \ref{#1 listing}, lines \ref{#2 begin}-\ref{#2 end}.)}}

Next, we will explain our algorithm for distributed computation of inductive signatures.
Pseudo code of the main loop can be found in Table \ref{refinement listing}.
It leaves out the details of the signature computation and global hash table.
These details can be found in table \ref{subroutine listing}.
The algorithm works in a few steps:
\begin{enumerate}\setlength{\itemsep}{-0.2em}
\item Put the initial partition (every state is equivalent) in the current partition and start the first iteration. \seelisting{refinement}{main init}
\item Initialize the data structure needed in each iteration.
That is, set the values of the successor partition IDs and signatures to undefined,
clear the global hash table, clear the signature and new ID queues and put all states
in the old ID queue.
\seelisting{refinement}{loop init}
\item If a state is in the old ID queue it means that the ID with respect to the previous partition
has to be forwarded to the predecessors. This is done by sending a message for every incoming
$A_?$ or $\tau$ edge. \seelisting{refinement}{old send}
If such a message is received then the old ID is stored and if necessary the state is put in the signature queue. \seelisting{refinement}{old recv}.
\item If a state is in the signature queue then all information needed to compute the signature is present.
Once the signature has been computed it is sent to all $\tau$ predecessors and a request is sent to the global hash table to resolve the ID of the (oldID, signature) pair. \seelisting{refinement}{sig}
If a signature set request is received then the signature is set and if necessary the state is put in the signature queue. \seelisting{refinement}{sig recv}
If a hash table request is received then the lookup is made and the reply is sent immediately.
\seelisting{refinement}{hash serve}
Upon receiving the reply, the state is put in the new ID queue. \seelisting{refinement}{hash recv}
\item If a state is in the new ID queue then the ID in the current partition is ready
to be sent to all $A_>$ predecessors. \seelisting{refinement}{new send}
Receiving such a message leads to storing the result and possibly inserting the state
in the signature queue. \seelisting{refinement}{new recv}
\item 
As soon as the new partition ID of every state is known everywhere, the message loop can exit.
Note that this requires a simple form of distributed termination detection.
\item By adding up the share of every partition ID hash table, we compute the number of partitions
and we repeat the loop if necessary.
\end{enumerate}

\begin{table}[tp]\caption{Subroutines for inductive branching minimization.}
\label{subroutine listing}
\hspace*{1.5em}
\begin{minipage}{\textwidth}
\begin{lstlisting}[numbers=left,language=pseudo]
proc check_ready($s$) 
   for $e$ in $succ(s)$ do
      if dest_id[$e$]=$\perp$ or $lbl(e)=\tau$ $\wedge$ dest_sig[$e$]=$\perp$ then return end
   end
   add $s$ to sig_queue
end
set compute_sig($s$)
    pre := $\emptyset$
    for $e$ in $succ(s) \cap E_w^{?}$ do pre := pre $\cup$ {($lbl(e)$,dst_old[$e$])} end
    for $e$ in $succ(s) \cap E_w^{>}$ do pre := pre $\cup$ {($lbl(e)$,dst_new[$e$])} end
    for $e$ in $succ(s)$ with $lbl(e)=\tau$ and dest_id[$s$] = dst_old[$e$] do
	if pre $\subseteq$ dest_sig[$e$]$\cup${($\tau$,dst_new[$e$])} then return dest_sig[e] end
    end
    return pre
end
int  index_count:=0;  hashtable index_table:=$\emptyset$
proc indexed_set_clear() index_count:=0; index_table:=$\emptyset$ end
int indexed_set_put(pair)
   if index_table[pair] = $\perp$ then
      index_table[pair]:=index_count*workers+me; index_count++ end
   return index_table[pair]
end
\end{lstlisting}
\end{minipage}
\end{table}

As described above, messages from the old queue, signature queue and new queue are dealt with in parallel
until finished. The actual implementation deals with these messages in waves:
first the entire old queue is dealt with then the signature queue and new queue are emptied globally in
sub iterations.

\smallskip

Before we discuss the experiments with our prototype implementation, we first discuss the
time, memory and message complexity. For this analysis we assume that the fan out of every state
is bounded. We assume an LTS with $N$ states and $M$ transitions.

The time needed for the algorithm is the number of iterations times the cost of each iteration.
The worst case number of iterations is the number of states $N$. (E.g. for the LTS
$( \{ 0 , \cdots , N-1 \} , i \arrowsuper{a} i+1 \mathrel{\rm mod} N \cup 0 \arrowsuper{b} 0 , 0 )$.)
In each iteration, for each state we must compute the signature and insert it in
the global hash table. Due to the fact that the fan out is constant, this requires
$\mathcal{O}(N)$ time and messages. For each edge, we may have to send the old ID,
the new ID and the signature. This requires $\mathcal{O}(M)$ time and messages.
Overall, the worst case time complexity is $\mathcal{O}(N \cdot {N+M})$.

The number of times one cannot avoid waiting for a message in each iteration depends
on the length of the longest $A_>$ path in the graph: computation has to start
at the last node and work up to the first, incurring three message latencies at each
step.

The memory needed by the algorithm to store the LTS and the signatures
is linear in the number of states and transitions: $\mathcal{O}(N+M)$.
(This is a difference to the old algorithm where even if the fan out was bounded,
the size of many signatures could be in the order of the number of edges.)
Provided that the owner functions work well, the memory use is evenly distributed across all workers.
The memory needed for message buffering can be kept constant, because each step that involves
sending more than one message is a step where a state has to be taken from a queue.
Blocking these steps if the number of messages in the system is above a threshold
limits the number of messages to that threshold.
Overall, the worst case memory complexity of the algorithm is $\mathcal{O}({N+M})$.

The worst case memory is also the expected memory complexity, since we expect to keep the LTS
in memory. The expected time complexity is much lower than the worst case:
The expected number of iterations and the expected length of the longest $A_>$ path
are orders of magnitude less than the number of states.


\section{Experimental Evaluation}
\label{experiments}

To study the performance of the implementation of the new algorithm,
we use four models. We perform two tests on these models.
First, we compare with
existing branching bisimulation reduction tools. Second, we test how well
the new implementation scales in the number of computes nodes and cores used per node.
In addition, we briefly mention work in progress on inductive strong bisimulation.

\smallskip

The models that we use in our experiments are:
\begin{description}\setlength{\itemsep}{-0.2em}
\item[lift6] A distributed lift system \cite{Groote2001a}.
This model describes a system that can lift large vehicles by using one leg for each wheel
of the vehicle. These legs are connected in a ring topology.
The instance we used has 6 legs.
\item[swp6] A version of the sliding window protocol \cite{DBLP:journals/fac/BadbanFGPP05}.
It has 2 data elements, the channels can contain at most one element and 
the window size is 6.
\item[fr53] A model of Franklin's leader election protocol for
anonymous processes along a bidirectional ring of asynchronous channels, which
terminates with probability one \cite{DBLP:conf/ifipTCS/BakhshiFPP08,DBLP:journals/cacm/Franklin82}.
We chose an instance with 5 nodes and 3 identities.
\item[1394fin] Model of the physical layer service
of the 1394 or firewire protocol and also the link
layer protocol entities \cite{SEN-R9706,DBLP:journals/sttt/SighireanuM98}.
We use an instance with 3 links and 1 data element.
\end{description}
The sizes of these models, in their original, cycle eliminated and branching reduced forms are shown in Table \ref{problem size table}. This table also show the number of iterations needed by classic branching (c.b.), inductive branching (i.b.), classic  strong (c.s.), inductive strong (i.s.) and
the length of the longest $\tau$ path (p).
Note that in two cases (lift6 and 1394fin) the number of iterations
needed by the inductive branching algorithm is less than the number needed by the classical algorithm.
Also note that the number of iterations needed for inductive strong bisimulation
is always a lot less. It will be interesting to see, if
we get similar results if we use real input graphs and $A_>$, instead of $\tau$-cycle reduced graphs
and $A_>=\{\tau\}$.

\begin{table}[tp]\caption{Problem sizes}\label{problem size table}
\resizebox{\textwidth}{!}{
\begin{tabular}{|l|r|r|r|r|r|r|r|r|r|r|r|}
\hline
& \multicolumn{2}{c|}{original}
& \multicolumn{2}{c|}{cycle free}
& \multicolumn{2}{c|}{branching}
& \multicolumn{5}{c|}{iterations}
\\\hline
 & states & trans. & states & trans. & states & trans. & c.b. & i.b. & c.s. & i.s. & p
\\\hline
lift6	& 33,949,609	& 165,318,222	& 33,946,699& 	165,312,102	& 12,463	& 71,466	& 16& 	8 & 91 & 7 & 78
\\\hline
swp6& 	56,793,060& 	271,366,320& 	13,606,212& 	56,996,856& 	8,191	& 16,380& 	13	& 13 & 20 & 13 & 51
\\\hline
1394fin& 	88,221,818& 	152,948,696	& 86,692,394& 	148,537,294	& 26,264	& 79,002	& 7	& 5 & 91& 6 & 75
\\\hline
fr53& 	84,381,157& 	401,681,445& 	81,115,587& 	385,379,715& 	2& 	1& 	2& 	2 &
 - & - & 196
\\\hline
\end{tabular}}
\end{table}

In Table \ref{comparison table}, we show the results of the comparison.
The tools in the comparison are
\begin{description}\setlength{\itemsep}{-0.2em}
\item[{bcg\underline{~}min}]
The reduction tool from the CADP toolset \cite{DBLP:conf/cav/GaravelMLS07}.
Version 1.7 from the 2007q beta release, 64 bit installation. This implements
the algorithm from \cite{DBLP:conf/icalp/GrooteV90}, for which first the $\tau$-cycles
must be eliminated (ce).
\item[ltsmin sequential]
The reduction tool which is released as part of the $\mu$CRL toolset
\cite{DBLP:conf/cav/BlomFGLLP01}. We additionally implemented a sequential version
of the inductive branching bisimulation algorithm in this tool.
\item[ltsmin distributed]
A distributed implementation, which contains the classic distributed branching 
bisimulation reduction algorithm from \cite{DBLP:journals/entcs/BlomO03},
and the newly implemented inductive branching bisimulation reduction algorithm.
\end{description}

For {\tt bcg\underline{~}min}, we show the total time needed for
reading the input, reducing and writing the output.
For ltsmin sequential, we show both the total time and the time needed for reduction.
For ltsmin distributed classic, we show the reduction time (wall clock time).
For ltsmin distributed inductive, we show the time for sequential cycle elimination
and the wall clock time of distributed reduction. In all cases we additionally show the total memory requirements in MB.
The tests were performed on a dual quad core Xeon 3GHz
machine with 48GB memory.

\begin{table}[tp]\caption{Sequential tool comparison.}\label{comparison table}
\resizebox{\textwidth}{!}{
\begin{tabular}{|l|r|r|r|r|r|r|r|r|r|r|r|r|r|r@{$+$}r|r|}
\hline
& \multicolumn{2}{c|}{bcg\underline{~}min}
& \multicolumn{9}{c|}{ltsmin (sequential implementation)}
& \multicolumn{5}{c|}{ltsmin (distributed, 4 cores)}
\\\hline
& \multicolumn{2}{c|}{ce + GV \cite{DBLP:conf/icalp/GrooteV90}}
& \multicolumn{3}{c|}{classic}
& \multicolumn{3}{c|}{ce + classic}
& \multicolumn{3}{c|}{ce + inductive}
& \multicolumn{2}{c|}{classic}
& \multicolumn{3}{c|}{ce + inductive}
\\\hline
 & time & mem & time & red & mem& time & red & mem & time & red & mem & red & mem  & \multicolumn{2}{c|}{red} & mem 
\\\hline\hline
lift6 & 1251	&6493& 261&225&2939&	298 & 261 &	2203	&191 & 154 &2299& 655 &7116	& 64 & 246&	5520
\\\hline
swp6 &	1298&	10699&342&287&5464&	264& 209 &	3625	&166&111&3573	&621&	12129&	73&133&	3587
\\\hline
1394 &		20906&	8226& 248&218&3473&	231 & 201	&2482&	144 & 114&2724	&730	&8657&62 & 272&	6315
\\\hline
fr53 &		204	&15870& 305	& 237 & 9744 & 1247 & 1180 &	5377	&715 & 651 &	5462	& 188&	16871	&624 & 476	&12991
\\\hline
\end{tabular}}
\end{table}

\smallskip

Several conclusions can be drawn from the results.
By looking at the results for sequential ltsmin, we can conclude that inductive signatures
are better than classic signatures. By looking at the times needed for fr53
it is obvious that this implementation of cycle elimination in ltsmin should
be improved.

We can also conclude that on these cases, sequential ltsmin uses much less memory than bcg\underline{~}min for branching bisimulation.
With the exception of fr53, sequential ltsmin is also much faster than bcg\underline{~}min.
Note that the differences in time/memory are partially due to differences in implementation.
For instance, bcgmin uses 64 bit pointers to represent partitions, whereas ltsmin uses 32 bit integers.

It is also clear that the distributed tool is much more expensive in time
and memory than the sequential tool. The extra cost in memory is easily explained.
In ltsmin, signature ID's are stored per state only. In ltsmin they have to be stored
per state and per transition. In ltsmin the LTS itself takes 4 bytes per state
and 8 bytes per transition (label and state). In ltsmin it takes 8 bytes per state and
24 bytes per transition (label, owner and state for ingoing and outgoing edges).
This mean that ltsmin has to work through roughly 3 times as much data in each iteration,
which might take up to 3 times as much time.
Frequent synchronization between the workers and having to send and receive information
that in ltsmin can simply be accessed is expected to account for a lot of time.

\begin{figure}[tp]
\begin{center}
\begin{tabular}{cc}
{\small swp 6} & {\small lift 6}
\\
\parbox{0.5\textwidth}{
\vspace*{-2em}
\hspace*{-2em}\scalebox{.28}{\includegraphics{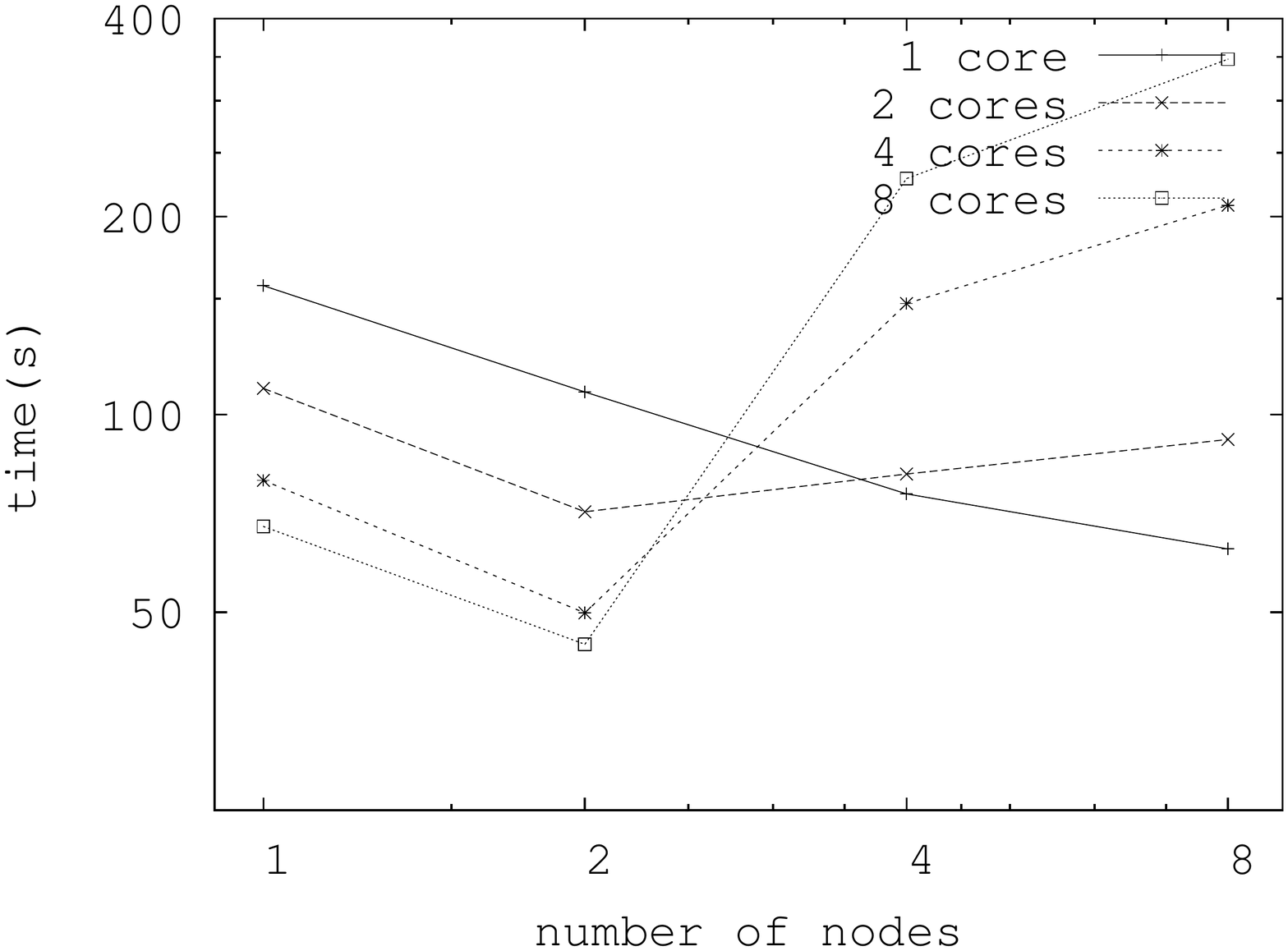}}\hspace*{-2em}
\vspace*{-0.7em}
}
&
\parbox{0.5\textwidth}{
\vspace*{-2em}
\hspace*{-2em}\scalebox{.28}{\includegraphics{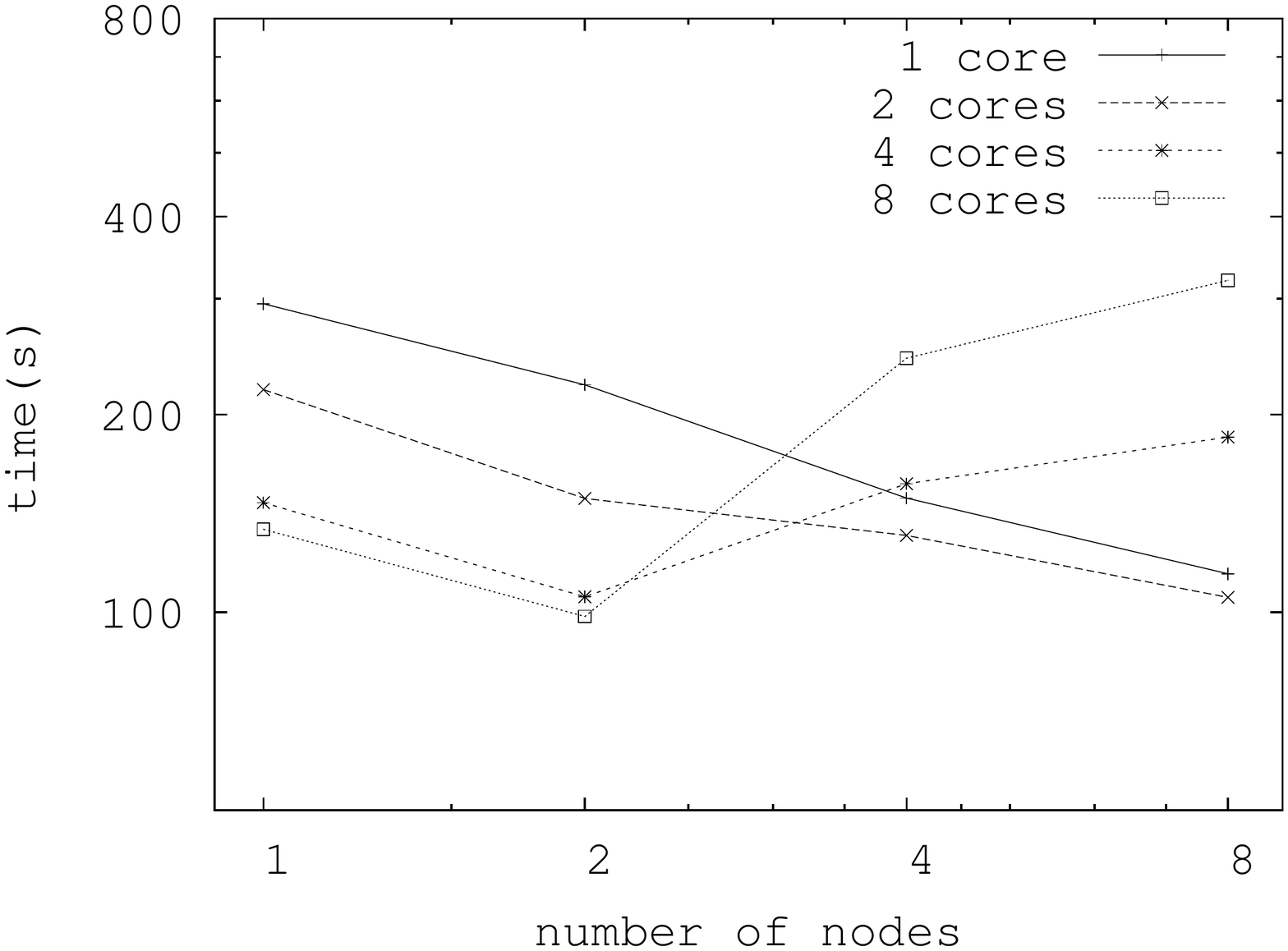}}\hspace*{-2em}
\vspace*{-0.7em}
}
\\
{\small 1394fin} & {\small franklin 5/3}
\\
\parbox{0.5\textwidth}{
\vspace*{-2em}
\hspace*{-2em}\scalebox{.28}{\includegraphics{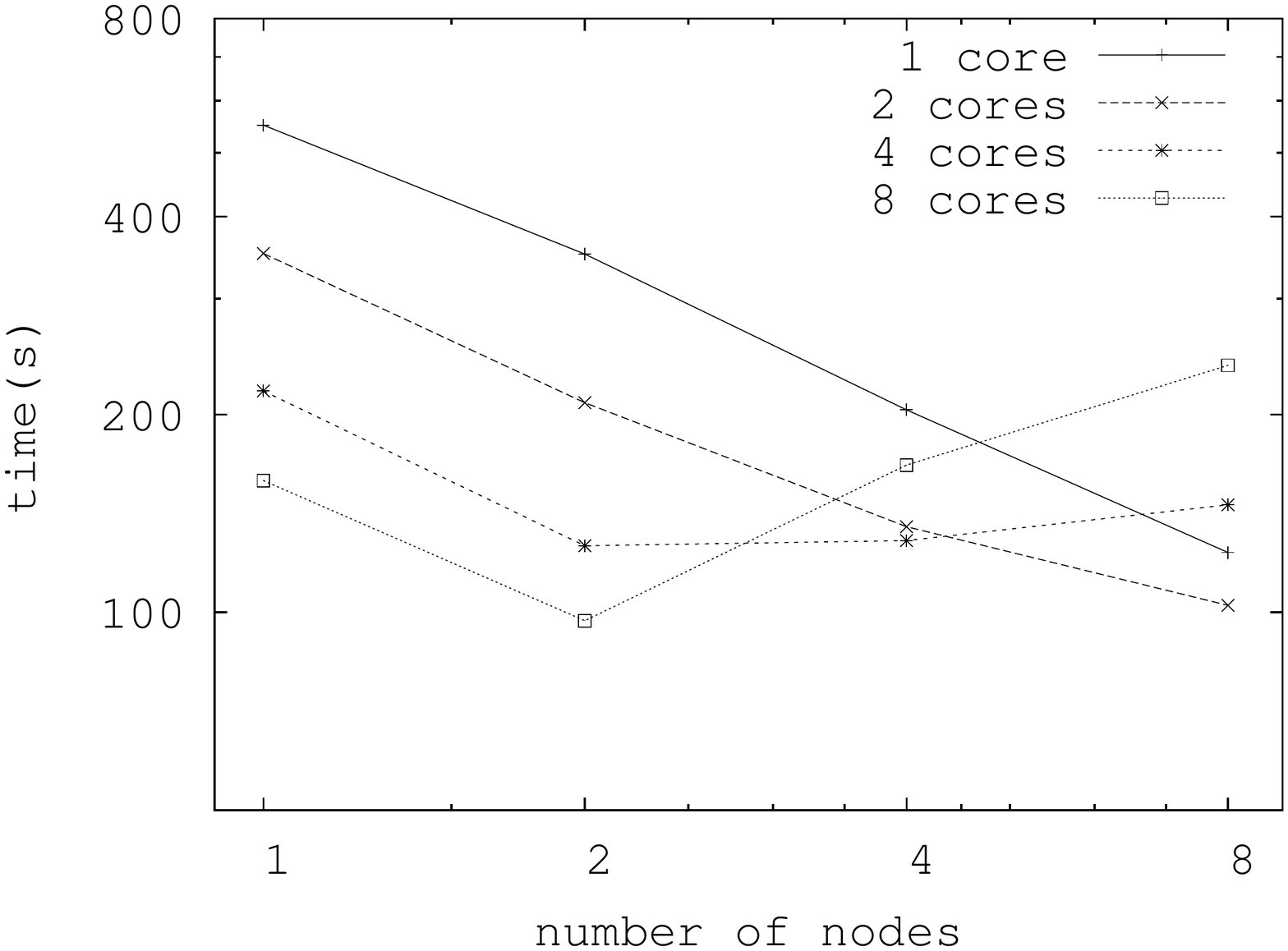}}\hspace*{-2em}
\vspace*{-1.5em}
}
&
\parbox{0.5\textwidth}{
\vspace*{-2em}
\hspace*{-2em}\scalebox{.28}{\includegraphics{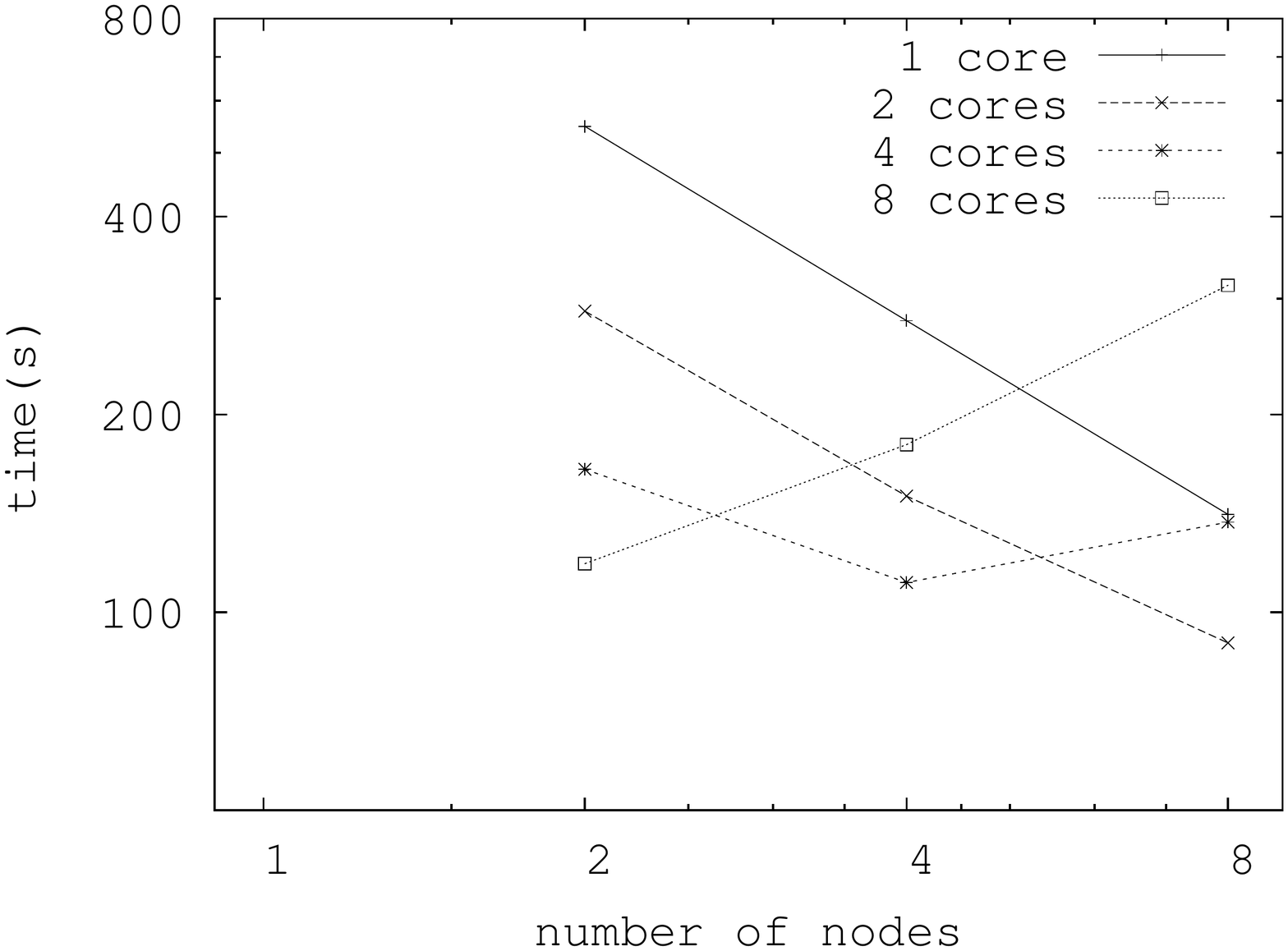}}\hspace*{-2em}
\vspace*{-1.5em}
}
\end{tabular}
\end{center}
\caption{Distributed reduction times for inductive branching bisimulation}
\label{reduction times figure}
\end{figure}

\smallskip

To test how well the algorithms scale, we first eliminated the $\tau$ cycles from
the four examples and then ran the inductive reduction on 1, 2, 4 and 8 nodes
with 1, 2 ,4 and 8 cores per node. For these tests, we used a cluster with
dual quad core Xeon 2GHz, 8GB memory machines connected with gigabit ethernet.
The times needed for the reduction can be seen in Fig. \ref{reduction times figure}.

The graphs have been ordered from the smallest to the largest problem.
It is interesting to see that for the smallest problem (swp6), the first time that 
more workers leads to more rather than less time is
using 2 nodes, 2 cores per node. For the next two (lift6,1394fin) this happens at 2 nodes,
4 cores per node and for the largest (franklin) at
4 nodes, 4 cores per node.

It is also clear that using 8 cores instead of 4 is problematic.
For 1 and 2 nodes the performance increase is small and for 4 and 8 nodes,
the performance actually gets worse.
Taken together with the huge difference in performance between the
sequential and the distributed tool this leads to the (unsurprising)
conclusion that it would be better to change the implementation to be aware
of which workers are local (allow shared memory) and which workers are remote
(require message passing). We leave such a tuned heterogeneous cluster-of-multi-cores
implementation for future work.

\section{Conclusion}
\label{conclusion section}

We have defined the notion of inductive branching signature
and proven that it corresponds to branching bisimulation.
We have given a distributed algorithm that computes
the coarsest branching bisimulation using inductive signatures.
In the experiments section, we have shown that it
is possible to implement the algorithm in such a way
that it scales for up to 8 workers with 1 or 2 cores.

The current prototype is good enough to show the merit of the concept
of inductive signatures. However, it can be optimized in several ways.
For example, the information about edges between two workers is currently stored
by both the source worker and the destination worker. If both workers are on the same
machine, then they could share a single instance of the data. Similarly,
the algorithm uses a lot of small messages. For good performance,
message combining is needed, which is currently done at the worker level, but could
be done at the node level instead.

Because strong bisimulation is a special case of branching bisimulation,
our algorithm can also be used for strong bisimulation. However,
for branching bisimulation we can eliminate $\tau$ cycles to
get a well-founded partition. For strong bisimulation, we will
have to come up with a good heuristic to automatically find
well-founded partitions.

As a final conclusion, we note that inductive signatures for branching
bisimulation improve time and memory requirements compared to
classical signatures, both in a sequential and a distributed
implementation. Of course, distributed minimization can handle larger
graphs that don't fit in the memory of a single machine. Additionally,
the distributed version using 8 cores on 2 nodes consistently beats
the best sequential algorithm in time.

\begin{pdmc}
\bibliographystyle{eptcs} 
\end{pdmc}
\begin{report}
\bibliographystyle{plain}
\end{report}
\bibliography{literature}

\end{document}